\documentclass{llncs}
\newif\ifproof
\prooftrue

\newif\ifprooftrivial
\prooftrivialfalse

\usepackage{booktabs} 

\usepackage[ruled]{algorithm}
\usepackage[noend]{algpseudocode}
\usepackage{color}
\usepackage{subfig}
\usepackage{epsfig}
\usepackage{epstopdf}
\usepackage{amsfonts}
\usepackage{latexsym}
\usepackage{amssymb}
\usepackage{amsmath}
\usepackage{fullpage}
\usepackage{thmtools}
\usepackage{thm-restate}

\newcommand{\includeFig}[4]      {\begin{figure}[htb] \begin{center}
\includegraphics[{#1}]{#3}\caption{\label{#2}#4} \end{center} \end{figure}}
                                 
\newcommand{\IR}{\mathbb{R}}
\newcommand{\BO}{\mathcal{O}}

\newcommand{\CR}{\textrm{CR}}
\newcommand{\CRa}{\widehat{\CR}}
\newcommand{\CRo}{\overline{\CR}}

\newcommand{\IS}{\mathcal{S}}
\newcommand{\IN}{\mathcal{N}}
\newcommand{\IF}{\mathcal{F}}
\newcommand{\fb}{F}
\newcommand{\Alg}{\mathcal{A}}
\newcommand{\OptAlg}{\overline{\Alg}}

\newcommand{\MEC}{\mathcal{MC}}
\newcommand{\Gather}[2]{\textit{Gather}($#1$, $#2$)}
\newcommand{\TriGather}{\Gather{3}{1} }
\newcommand{\NGather}{\Gather{n}{1} }
\newcommand{\Rad}{\mbox{Radius}}
\newcommand{\Center}{\mbox{Center}}
\newcommand{\Sup}{\mbox{Sup}}
\newcommand{\edge}[1]{\overline{#1}}

\newcommand{\FVD}{\mathcal{FVD}}
\newcommand{\FVR}{\mathcal{FVR}}

\begin{document}

\setlength{\abovecaptionskip}{5pt}

\title{Gathering in the Plane of Location-Aware Robots in the Presence of Spies}
  
\author{
    Jurek Czyzowicz\inst{1}\inst{4}
    \and
    Ryan Killick\inst{2}
    \and
    Evangelos Kranakis\inst{2}\inst{5}
    \and
    Danny Krizanc\inst{3}
    \and
    Oscar Morale-Ponce\inst{4}
    }
    
    \institute{
    D\'{e}partemant d'informatique, Universit\'{e} du Qu\'{e}bec en Outaouais,  
    Canada.
    \and
    School of Computer Science, Carleton University, Ottawa, Ontario, Canada.
    \and
    Department of Mathematics \& Computer Science, Wesleyan University,
    Middletown CT, USA.
    \and
    Department of Computer Science, California State University,
    Long Beach, CA, USA.
    \and
    Research supported in part by NSERC Discovery grant.
    }

\maketitle

\begin{abstract}
    A set of mobile robots (represented as points) is distributed in the Cartesian plane. The collection contains an unknown subset of byzantine robots which are indistinguishable from the reliable ones. The reliable robots need to gather, i.e., arrive to a configuration in which at the same time, all of them occupy the same point on the plane. The robots are equipped with GPS devices and at the beginning of the gathering process they communicate the Cartesian coordinates of their respective positions to the central authority. On the basis of this information, without the knowledge of which robots are faulty, the central authority designs a trajectory for every robot. The central authority aims to provide the trajectories which result in the shortest possible gathering time of the healthy robots. The efficiency of a gathering strategy is measured by its competitive ratio, i.e., the maximal ratio between the time required for gathering achieved by the given trajectories and the optimal time required for gathering in the offline case, i.e., when the faulty robots are known to the central authority in advance. The role of the byzantine robots, controlled by the adversary, is to act so that the gathering is delayed and the resulting competitive ratio is maximized.  
    
 The objective of our paper is to propose efficient algorithms when the central authority is aware of an upper bound on the number of byzantine robots.  We give optimal algorithms for collections of robots known to contain at most one faulty robot. 
    When the proportion of byzantine robots is known to be less than one half or one third, we provide algorithms with small constant competitive ratios. We also propose algorithms with bounded competitive ratio in the case where the proportion of faulty robots is arbitrary. 
\end{abstract}
\section{Introduction}

\subsection{The background}
    A collection of mobile robots need to meet at some point of the geometric environment. This task, known as {\em gathering} or {\em rendezvous}, has been extensively investigated in the past. The gathering may be necessary, e.g., to coordinate a future task or to exchange previously acquired information.

    In most formerly studied cases, robots have limited knowledge about the environment and they do not know the positions of the other robots. 
    In the present paper, the robots are distributed in the two-dimensional Cartesian plane. They are equipped with GPS devices and they can wirelessly communicate their positions to the central authority. The central authority then informs each individual robot of the trajectory it is to follow in order to meet. However, the team of reliable robots has been contaminated with ``spies" - a subset of byzantine robots, indistinguishable from the original ones,  controlled by an omnipotent adversary. The role of the faulty robots is simple -- delay the gathering of the reliable ones for as long as possible. A byzantine robot may report a wrong position, fail to report any, or fail to follow its assigned route. As the central authority does not recognize which robots are byzantine, it sends the travel instructions to all of them.

    Our goal is to design a strategy resulting in gathering of all reliable robots within the smallest possible time. We attempt to minimize the \textit{competitive ratio} -- the ratio of the time required to achieve gathering of the reliable robots, to the time required for such gathering to occur under the assumption that the reliable robots were known in advance.

\subsection{The model and the problem}
    A collection $\IS$ of $n$ mobile robots move at maximum unit speed within the two-dimensional plane. It is assumed that each robot in $\IS$ is equipped with a GPS device so it is aware of a pair of Cartesian coordinates representing its current location in the plane. 
    
    We consider the problem of gathering an unknown subset $\IN \subseteq \IS$ of robots. The robots of $\IN$ need to arrive at some time at a same point on the plane in order to complete some given task.
 We refer to this set $\IN$ of at least $n-\fb$ robots as the set of reliable robots and define $\IF = \IS \setminus \IN$ of $f \leq \fb$ robots as the set of byzantine robots. We call this problem of gathering all reliable robots from a collection containing at most $F$ byzantine robots the \Gather{n}{\fb} problem.

    At the beginning, all robots in $\IS$ send a single message recording their starting positions to the central authority. In turn, the central authority computes a set of trajectories instructing each robot how to time their respective movements in order to achieve gathering. At this point the robots follow the trajectories provided.
    
    The movement continues until all reliable robots meet for the first time. 
    We imagine a successful gathering as a meeting of robots possessing pieces of information allowing them to solve some puzzle. As long as all pieces are disassembled, the puzzle remains unsolved, and the identification of useful or invalid information is not possible.

    The byzantine robots may report incorrect initial locations, which can potentially adversely affect the robots' trajectories. Clearly, this results in byzantine robots not being able to follow the assigned trajectories. However, as long as all reliable robots complete their trajectories, the schedule must lead to their gathering. 
    
    The trajectories designed by the central authority are computed uniquely on the basis of the reported set of robot positions and possibly using the knowledge of the upper bound on the number of byzantine robots. Once the robots start their movements, no adaptation to our algorithm is ever possible as no extra information may be obtained. We assume that the adversary knows in advance our algorithm and it will put the byzantine robots in the positions which result in the worst possible competitive ratio.
    
    We note that the requirement of a central authority may be removed by allowing the robots to instead broadcast their initial positions to all other robots. In this situation all robots compute the same set of trajectories using the same algorithms.


    We are interested in developing algorithms solving the \Gather{n}{\fb} problem which are optimal in terms of the competitive ratio for a given initial configuration $\IS$ of $n$ robots, at most $\fb$ of which are byzantine. We define the competitive ratio $\CR_{n,\fb}(\Alg, \IN)$ of an algorithm $\Alg$ for the specific subset $\IN$ of the input $\IS$ as the ratio of the time $T_{\Alg}(\IN)$ -- the time of the first gathering of all robots belonging to $\IN$ -- divided by $T_*(\IN)$ -- the minimal time necessary to gather the robots in $\IN$, i.e. 
        $
        \CR_{n,\fb}(\Alg, \IN) = \frac{T_{\Alg}(\IN)}{T_*(\IN)}.
        $
    We also define the \textit{overall} competitive ratio $\CRa_{n,\fb}(\Alg, \IS)$ of an algorithm $\Alg$ with input $\IS$ as the maximal $\CR_{n,\fb}$ over any subset $\IN$ of $\IS$, i.e.
        $
        \CRa_{n,\fb}(\Alg, \IS) = \max_{\IN \subset \IS} \CR_{n,\fb}(\Alg, \IN).
        $
    We further define the \textit{optimal} competitive ratio $\CRo_{n,\fb}(\IS)$ for an input $\IS$ as the minimal $\CRa_{n,\fb}(\Alg, \IS)$ for any algorithm $\Alg$, i.e.
        $
        \CRo_{n,\fb}(\IS) = \min_{\Alg} \CRa_{n,\fb}(\Alg, \IS).
        $
    For ease of presentation we will often drop the subscripts $n$ and $F$ when they are implied by context.

    We define an optimal algorithm $\OptAlg$ solving the \Gather{n}{\fb} problem as any algorithm satisfying
    \begin{equation}\label{eq:opt_alg}
    \CR_{n,\fb} (\OptAlg, \IS) = \CRo_{n,\fb}(\IS),\  \forall\ \IS.
    \end{equation}


\subsection{Our results}
    We provide algorithms with constant competitive ratio for all but a small bounded region in the space of possible $n$ and $F$ pairs. In doing so we demonstrate that having knowledge of the upper bound of the number of byzantine robots in the subset (represented by the parameter $\fb$) permits fine-tuning of the gathering algorithm, resulting in better competitive ratios.

    In Section 2 we consider the gathering problem for collections involving only a single byzantine robot. After developing insight into the problem we give a gathering algorithm that is optimal for any number of robots, at most one of which is byzantine. For the boundary case of three robots, one of which is byzantine, we give a closed form expression for the competitive ratio.

    Section 3 presents two algorithms with small constant competitive ratio when the number of byzantine robots is bounded by a small fraction of $n$. Specifically, we give algorthms with competitive ratios of 2 and $2\sqrt{2}$ when $\fb < \lceil n/3 \rceil$ and  $\fb < \lceil n/2 \rceil$ respectively.

    Finally, in Section 4, we give two gathering algorithms solving the problem for any $n$ and any $\fb$. The competitive ratio of one of these algorithms is constant, while the other is bounded by $\fb+2$.

    We summarize the results of the paper in Table~\ref{tab:n_f} and Figure~\ref{fig:n_f}.
    \vspace{-0.1cm}
    \begin{center}
        \begin{minipage}[t][3.8cm]{.5\linewidth}
        \vspace{0pt}
        \centering
        \captionof{table}{Summary of competitive ratio bounds for various algorithms.}
        \begin{tabular}{|c | c | c |}
            \hline
            $F$                               & Upper-bound   & Reference\\
            \hline\hline
            1                                 & optimal       & Alg.~\ref{alg:opt_n} \\
            $\leq \lceil n/3 \rceil $         & 2             & Alg.~\ref{alg:cp_n3} \\
            $\leq \lceil n/2 \rceil $         & $2\sqrt{2}$   & Alg.~\ref{alg:cp_n2} \\
            $>  \lfloor32\sqrt{2}\rfloor-2$   & $32\sqrt{2}$  & Alg.~\ref{alg:gr} \\
            $\leq \lfloor32\sqrt{2}\rfloor-2$ & $\fb+2$       & Alg.~\ref{alg:ssi} \\
            \hline
        \end{tabular}
        \label{tab:n_f}
        \end{minipage}%
        \begin{minipage}[t][3.8cm]{.5\linewidth}
        \vspace{0pt}
        \centering
        \includegraphics[width=5cm, keepaspectratio]{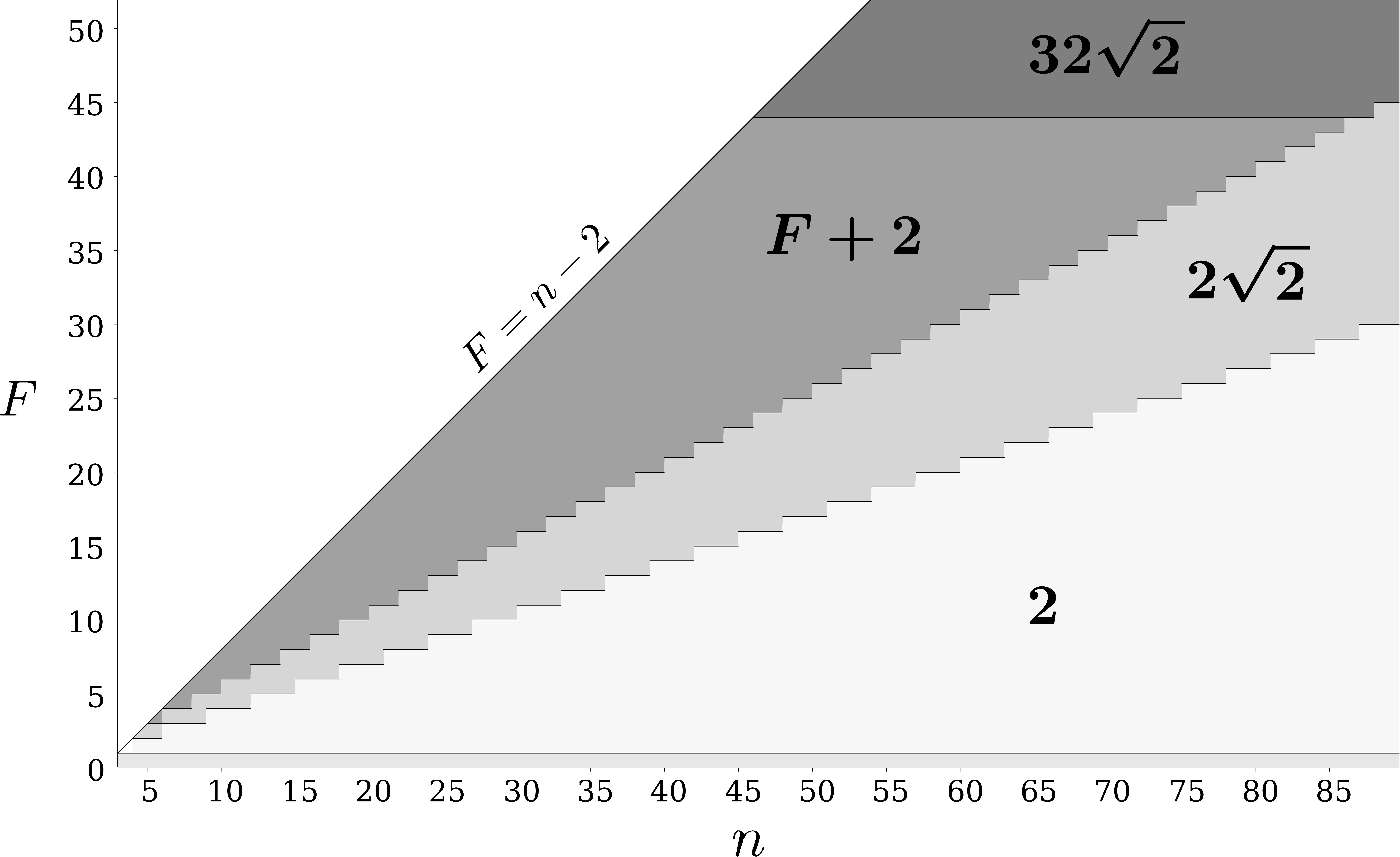}
        \captionof{figure}{Competitive ratio bounds for various regions of the space of possible $n$ and $F$ pairs.}
        \label{fig:n_f}
        \end{minipage}
        \end{center}
\vspace{0.2cm}
\subsection{Related work}
    The gathering problem was originally introduced in \cite{suzuki1999distributed} as a version of pattern formation (see also \cite{Das10}).  In operations research, Alpern \cite{alpern1995rendezvous,AlpernPerspective} considers the gathering of two robots, 
    referred to as the {\em rendezvous} problem (cf. \cite{Pelc11}). 
    Both problems are central in theoretical computer science. The rich related literature is due to the large variety of studied settings: deterministic and randomized, synchronous and asynchronous, for labeled and anonymous agents, in graphs and geometric environments, for same-speed or distinct maximal speed agents, etc. (cf. \cite{alpern2003theory,Cieliebak03,collins2010tell,marco2006asynchronous,Flocchini01,Kranakis06,Prencipe07,YuDynamic}). More recently, efficient solutions were proposed for the plane \cite{CourtieuRTU16} and for grids \cite{SPAA16}.

    In many papers on gathering the agents are a priori assumed to have limited knowledge of the environment. Moreover, most papers supposed that an agent is not aware of the positions in the environment of other agents. In the deterministic settings, one of the central studied questions was feasibility of gathering or rendezvous, cf. \cite{Das10,marco2006asynchronous,Prencipe07}, which most often led to some form of the symmetry breaking problem, see \cite{Kranakis06,Pelc11}. Surprisingly, when agents were equipped with GPS devices, knowledge of the agent's own position in the environment permitted executing very efficient rendezvous algorithms (see \cite{collins2011synchronous,collins2010tell}). 

    Fault tolerance in mobile agent algorithms has also been extensively studied in the past, but the failures were more often related to the static elements of the environment (network nodes or links), cf. \cite{hromkovic,lynch}. The faults  of the mobile agents were studied for the problems of convergence \cite{CohenP05}, flocking \cite{YangSDT11}, searching \cite{ISAAC16,PODC16} or patrolling \cite{ISAAC15}. Faults or imperfections arriving to mobile agents performing gathering were investigated in \cite{agmon2006fault,cohen-peleg-inaccurate-sensors,dieudonne2014gathering,izumi-unreliable-compasses,souissi2006gathering}. Research in \cite{cohen-peleg-inaccurate-sensors},  \cite{izumi-unreliable-compasses} and \cite{souissi2006gathering} considered the gathering problem in the presence of inaccurate or faulty robot perception components. In \cite{agmon2006fault} the initial positions of the collection is known to all robots, which operate in so called {\em look-compute-move} cycle. The feasibility of the problem, as a function of faulty robots, is investigated in \cite{agmon2006fault} for crash and byzantine faults. In \cite{dieudonne2014gathering}, the gathering problem is studied in an unknown graph environment and the feasibility question for byzantine faults in the strong and weak sense are investigated. The results of \cite{dieudonne2014gathering} depend on the knowledge of the upper bound on the size of the graph environment (or the absence of such knowledge).

    In \cite{hoda2017_2} and \cite{hoda2017} the authors studied, similar to ours, the online rendezvous problem using GPS-equipped robots on a line, where some robots may turn out to be byzantine. However the robot movements along the line are much easier to analyze than the setting studied in the present paper. Indeed, in the case of a line, the robots move inside a corridor forcing robots to meet. 

\subsection{Notation}
    We will use $\IS$ to refer to a general collection of any robots (reliable and/or byzantine) and use $\IN$ ($\IF$) to represent a set of reliable (byzantine) robots only. We will represent the cardinality of a set $\IS$ as $|\IS|$ and will always use $n = |\IS|$, and $f = |\IF|$. We reserve the use of $\fb$ for the upper bound on the number of byzantine robots in $\IS$ (and, as such, it may be that $f \leq \fb$).
    
    As we are dealing with robots in the plane we will use the term robot and point interchangeably. When it is required to refer to a particular robot / robots in a set we will use the capital letters $A$, $B$, and $C$. We use the capital letter $D$ to refer to meeting points of robots.

    We let the distance between any two points $A$ and $B$ be $|AB|$, and use $\edge{AB}$ to represent the directed line segment joining $A$ and $B$. 
    We will refer to the individual coordinates of a point using the subscripts $x$ and $y$, e.g., $A = (A_x, A_y)$.

    We define $\MEC(\IS)$ as the minimum enclosing circle (MEC) of a set of points $\IS$, and let $\Sup[\IS]$ be the supporting set of $\MEC(\IS)$. It is a well known property that $2\leq |\Sup[\IS]| \leq 3$ \cite{chrystal1885problem}. We further define the radius $\Rad[\IS]$ and $\Center[\IS]$ of $\IS$ to be the radius and center of the MEC of $\IS$ respectively.

    Finally, we let $\FVD(\IS)$ represent the furthest-point Voronoi diagram (FVD) of the point set $\IS$, and, for a point $A$ in $\IS$, we let $\FVR(A)$ be the cell / region in $\FVD(\IS)$ belonging to the point $A$. 
    See \cite{Berg:2008:CGA:1370949} for a description of the properties of the FVD.

\section{One byzantine robot} \label{sec:one_fault}

In this section we develop optimal algorithms for the case that there is only a single byzantine robot within the collection $\IS$. 

To do this we will need to consider subsets of $\IS$ containing $n-1$ robots and we therefore introduce some convenient notation. We let $\IS_i \subset \IS$, $i\in[0,n-1]$ represent the $n$ subsets of $n-1$ robots that can be formed from $\IS$ and we define an ordering for the $\IS_i$ in such a way that $\Rad[\IS_i] \leq \Rad[\IS_j]\ \forall\ j\geq i$. For the sake of brevity, we use $r_{\IS} = \Rad[\IS]$ and $r_i = \Rad[\IS_i]$ for the remainder of the section.

We start with the following (trivial) lemma concerning the optimal meeting time of any set of robots in the plane,

    \begin{lemma}\label{lm:t_min}
    The minimal time necessary to gather any set $\IS$ of robots is $T_*(\IS) = r_{\IS}$.
    \end{lemma}


An immediate consequence of the above lemma is the following optimal algorithm for gathering a group of $n$ reliable robots.

\vspace{-0.5cm}
    \begin{algorithm}[H] \caption{(Optimal \Gather{n}{0})} \label{alg:opt_nn}
    \begin{algorithmic}[1]
    \State{Set $D = \Center[\IS]$ ;}
    \State{All robots in $\IS$ move at full speed towards $D$ ;}
    \State{The algorithm terminates when the last robot in $\IS$ reaches $D$ ;}
    \end{algorithmic}
    \end{algorithm}
\vspace{-0.5cm}
To get an idea of how different the problem is when we consider the presence of even a single byzantine robot, let us run the above algorithm on the two inputs depicted in Figure~\ref{fig:intro_fig}.


\begin{center}
    \vspace{-0.2cm}
    \begin{minipage}[t][4.5cm]{.55\linewidth}
    \vspace{0pt}
    \centering
    \includegraphics[width=8.5cm,keepaspectratio]{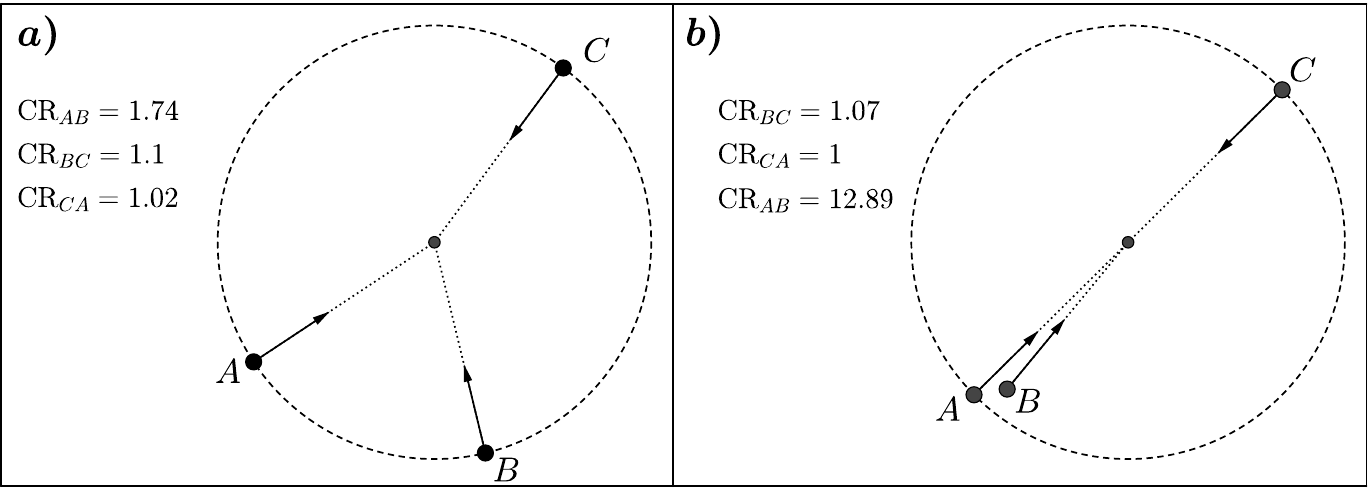}
    \label{fig:intro_fig}
    \captionof{figure}{Inputs for example analysis of competitive ratio. In both cases the robots $A$, $B$, and $C$ move directly towards the center of the minimum enclosing circle of $\IS=\{A,B,C\}$.}
    \end{minipage}%
    \begin{minipage}[t][4.5cm]{.45\linewidth}
    \vspace{0pt}
    \centering
    \includegraphics[width=4cm,keepaspectratio]{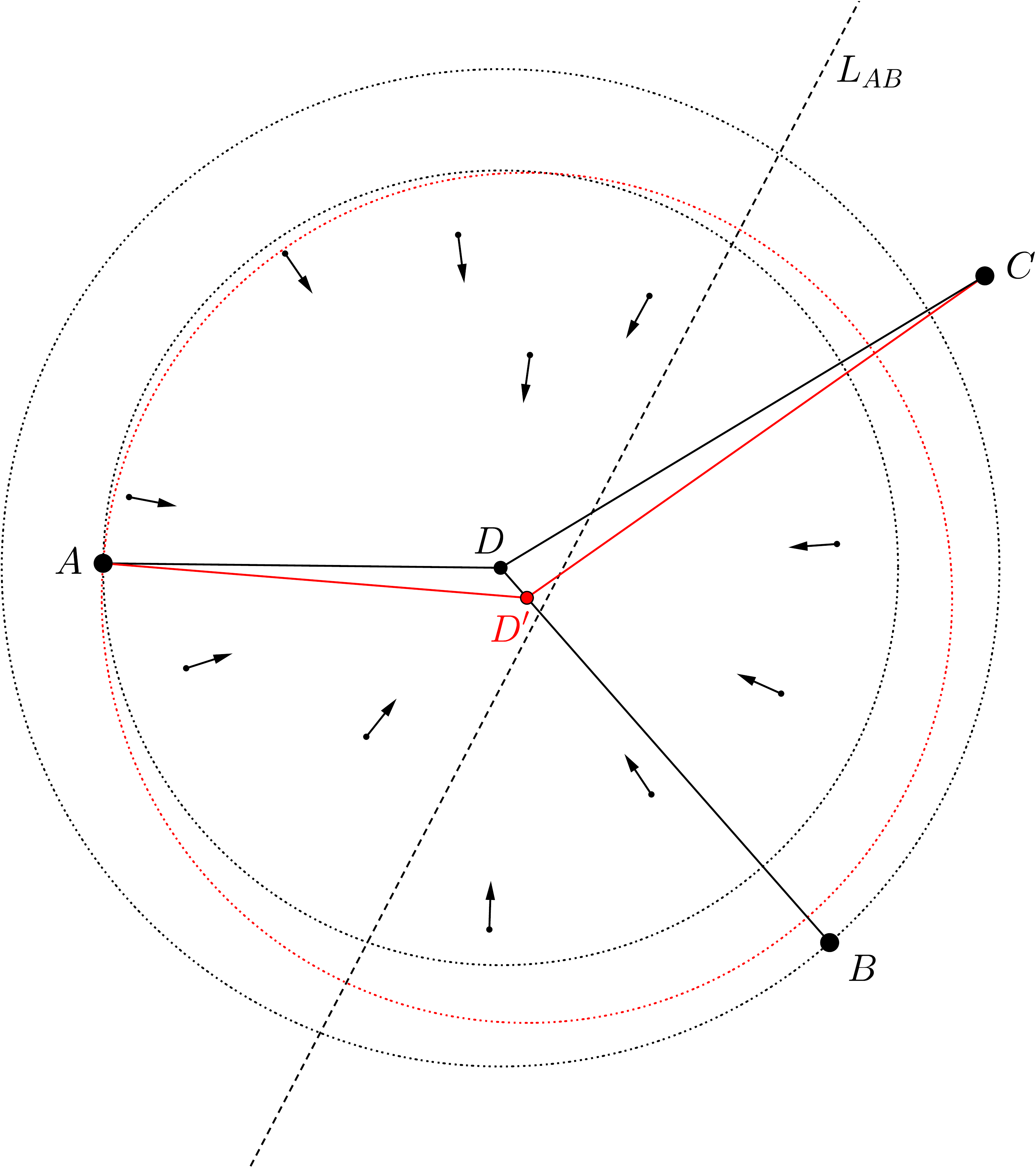}
    \captionof{figure}{Setup for the proof of Lemma~\ref{lm:bisect_n}.}
    \label{fig:bisect_n}
    \end{minipage}
    \end{center}

For a given input $\IS=\{A,B,C\}$ the adversary can choose at most one of the robots $A$, $B$, and $C$ to be byzantine. We assume that they will do so in such a way as to maximize the competitive ratio of our algorithm. Which robot would they choose? 

In the case \textbf{\textit{a})} the choice is not so obvious, and, indeed, the competitive ratios for all three possibilities are not very different. In the case \textbf{\textit{b})}, however, there \textit{is} an obvious choice: the adversary would make $C$ byzantine since the robots $A$ and $B$ were initially very close but travelled far before meeting.

This exercise, although simple, highlights an important observation -- the ``closest'' robots should meet first. It turns out that, when $F=1$, we can formalize this statement\footnote{When $F>1$ there are cases when this is not true.}.

    \begin{lemma}\label{lm:closest_n}
    Consider an optimal algorithm $\OptAlg$ solving the \NGather problem for the input $\IS$. Let $\IS_i$ be the first group of $n-1$ robots to meet. Then $\IS_i = \IS_0$, i.e. $\IS_i$ is the group of $n-1$ robots in $\IS$ with the smallest enclosing circle.
    \end{lemma}

    \ifproof
        \begin{proof}(Lemma~\ref{lm:closest_n}) Assume we have an optimal algorithm $\OptAlg$ such that $\IS_i$ -- the first group to gather using $\OptAlg$ -- is not the group with the smallest enclosing circle, i.e. $\IS_i \neq \IS$. In this case an adversary chooses $\IN = \IS_0$. Since the minimal time at which all robots can gather is $r_{\IS}$, the competitive ratio of $\OptAlg$ is $\geq r_{\IS}/r_0$.

        Now apply Algorithm~\ref{alg:opt_nn} to solve this \NGather problem and observe that the competitive ratio of this algorithm is {\em equal} to $r_{\IS}/r_0$. This implies that $\OptAlg$ is, at best, as good as Algorithm~\ref{alg:opt_nn}. However, Algorithm~\ref{alg:opt_nn} is not an optimal algorithm solving the \NGather problem. Thus, we must conclude that $\OptAlg$ is not an optimal algorithm either -- a contradiction.
        \qed \end{proof}
    \fi

So, we now know that we have to make the smallest group of $n-1$ robots meet first. What choice does this leave the adversary? Well, naturally, they would choose the byzantine robots in such a way that the second-smallest group of $n-1$ robots should have gathered. This observation leads us to the following:

    \begin{theorem}\label{thm:n_lb}
    The competitive ratio of any algorithm solving the \NGather problem with input $\IS$ is at least $r_{\IS}/r_1$.
    \end{theorem}

    \ifproof
        \begin{proof}(Theorem~\ref{thm:n_lb}) Consider an algorithm $\Alg$ solving the \NGather problem with input $\IS$. Let $\IS_i$ be the first group of $n-1$ robots to meet using $\Alg$, and let $\IS_j = \IN$ be the group of $n-1$ reliable robots. Observe that an adversary can always choose to make $\IS_i \neq \IS_j$ such that, effectively, all $n$ robots must meet before $\Alg$ terminates. Let the time at which this happens be $T$. Lemma~\ref{lm:t_min} tells us that $r_{\IS}$ is the minimum time necessary to gather all $n$ robots and we thus have $T\geq r_{\IS}$. The competitive ratio of $\Alg$ is therefore at least $\CRa \geq r_{\IS}/r_j,\ j\neq i$. There are two cases to consider: $i=0$ and the adversary chooses $j = 1$ such that $\CR \geq \frac{r_{\IS}}{r_1}$, or, $i\neq0$ and the adversary chooses $j = 0$ such that $\CR \geq \frac{r_{\IS}}{r_0} \geq \frac{r_{\IS}}{r_1}$. 
        \qed \end{proof}
    \fi

At this point we can make a useful observation: an optimal gathering algorithm ends either at the moment the first group of robots meet or the moment all robots meet. Furthermore, at the moment of the first meeting, all robots are located at either one of only two positions. Thus, in an optimal algorithm, we must send these remaining two groups of robots directly towards each other. We can claim the following:

    \begin{lemma}\label{lm:one_point}
    An optimal algorithm $\OptAlg$ solving the \NGather problem can be completely described by the single point $D$ at which the first $n-1$ robots gather.
    \end{lemma}


    \begin{corollary}\label{cor:one_point}(Lemma~\ref{lm:one_point})
    There is an optimal algorithm solving the \NGather problem following the strategy given in Algorithm~\ref{alg:general_n}.
    \end{corollary}

    \vspace{-0.6cm}
    \begin{algorithm}[H] \caption{(General \NGather)} \label{alg:general_n}
    \begin{algorithmic}[1]
    \State{All $n$ robots start moving at full speed towards some point $D$ ;}
    \If{The first $n-1$ robots to arrive at $D$ are all reliable;}
    \State{The algorithm terminates ;}
    \Else
    \State{Let $D'$ be the midpoint of $D$ and the position of the single robot that has not yet arrived at $D$ (at the time the first group of robots gather at $D$) ;}
    \State{All robots move at full speed towards $D'$. The algorithm terminates once they meet ;}
    \EndIf
    \end{algorithmic}
    \end{algorithm}
    \vspace{-0.5cm}

Corollary~\ref{cor:one_point} reduces the task of searching for an optimal algorithm to the conceptually simpler task of searching for some optimal meeting point $D$. The following lemma tells us how to find this point:

    \begin{lemma}\label{lm:bisect_n}
    Consider an optimal algorithm $\OptAlg$ solving the \NGather problem for the input $\IS$ parameterized by the point $D$. Let the group $\IS_i$ represent the first group of $n-1$ robots to gather at the point $D$. Then the point $D$ lies on the perpendicular bisector of the two robots in $\IS_i$ furthest from $D$.
    \end{lemma}

    \ifproof
        \begin{proof}(Lemma~\ref{lm:bisect_n})
        Let $A$ and $B$ be the last two robots in $\IS_i$ that reach $D$ and let $C$ be the single robot in $\IS$ that is not contained in $\IS_i$. We argue by contradiction and assume that the point $D$ does not lie on the perpendicular bisector $L_{AB}$ of $A$ and $B$. Without loss of generality assume that $A$ reaches $D$ before $B$. This situation is depicted in Figure~\ref{fig:bisect_n}. 
        

        If we let the group of $n-1$ reliable robots be $\IS_j = \IN$ then the competitive ratio of $\OptAlg$ is
        \[\CRa(\OptAlg, \IS) = \max_{S_j} \begin{cases}
            \frac{1}{r_i}|BD|,& \IS_i = \IS_j \\ 
            \frac{1}{2r_j}(|BD|+|CD|),& \mbox{otherwise}
        \end{cases}.\]
        
        Now consider the algorithm $\Alg'$ which replaces the point $D$ in $\OptAlg$ with the point $D'$ on the segment $\edge{BD}$ located some small distance $\epsilon$ in the direction of $L_{AB}$ such that $A$ and $B$ are still the last two robots to arrive at $D'$ (and $A$ arrives before $B$). The competitive ratio of this new algorithm is,
        \[\CRa(\Alg', \IS) = \max_{S_j} \begin{cases}
            \frac{1}{r_i}|BD'|,& \IS_i = \IS_j \\ 
            \frac{1}{2r_j}(|BD'|+|CD'|),& \mbox{otherwise}
        \end{cases}.\]
        We claim that $\CRa(\Alg', \IS) < \CRa(\OptAlg, \IS)$. It is obvious that $|BD'| < |BD|$. Thus we need to show that $|BD'|+|CD'| < |BD|+|CD|$. Indeed, observe that
        $|BD|+|CD| = |BD'|+|DD'|+|CD| > |BD'|+|CD'|\ (\mbox{triangle inequality}).$
        We can thus conclude that $\CRa(\Alg', \IS) < \CRa(\OptAlg, \IS)$ which is in contradiction to our assumption that $\OptAlg$ is an optimal algorithm.
        \qed \end{proof}
    \fi

As a last step we derive an expression for the competitive ratio of an optimal \NGather algorithm.

    \begin{lemma}\label{lm:n_cr}
    An optimal algorithm following the strategy in Algorithm~\ref{alg:general_n} solves the \NGather problem for the input $\IS$ with competitive ratio
    $
    \CRa = \max \left \{ \frac{|AD|}{r_0},\ \frac{|AD|+|CD|}{2r_1} \right\}
    $
    where $A$ is one of the two points in $\IS_0$ furthest from $D$ and $C$ is the point in $\IS$ that is not in $\IS_0$.
    \end{lemma}

    \ifproof
        \begin{proof}(Lemma~\ref{lm:n_cr}) Lemma~\ref{lm:closest_n} tells us that the first group of $n-1$ robots to gather is the group $\IS_0$. Thus, if $A$ is the point in $\IS_0$ that is furthest from $D$, then, in the case that $\IS_0 = \IN$, the competitive ratio of the algorithm is $|AD|/r_0$. 
        
        If $\IS_0 \neq \IN$, then, if $C$ is the single point in $\IS$ that is not in $\IS_0$, the algorithm terminates after time $(|AD|+|CD|)/2$. Thus, the overall competitive ratio of the algorithm is
        $\CRa(D) = \max \left \{ \frac{|AD|}{r_0},\ \frac{|AD|+|CD|}{2r_1} \right\}$
        as required.
        \qed \end{proof}
    \fi

We are now ready to present our main result:

\vspace{-0.5cm}
    \begin{algorithm}[H] \caption{(Optimal \NGather point)} \label{alg:opt_npoint}
    \begin{algorithmic}[1]
    \State{Set $C$ as the single robot in $\IS$ that is not in $\IS_0$;}
    \State{Determine the Furthest-point Voronoi diagram $\FVD(\IS_0)$ of the point set $\IS_0$;}
    \State{Set $CR_{min} = \infty$, and $D_{min} = NULL$;}
    \For{each edge $E$ in $\FVD[\IS_0]$}
    \State{Set $A$ and $B$ as the two points such that the edge $E$ separates $\FVR(A)$ and $\FVR(B)$;}
    \State{Determine the point $D'$ on $E$ that minimizes
        $\CR(D') = \max \left \{ \frac{|AD'|}{r_0},\ \frac{|AD'|+|CD'|}{2r_1} \right\}.$}
    \If{$\CR(D') < \CR_{min}$}
    \State{Set $D_{min} = D'$ and $\CR_{min} = \CR(D')$}
    \EndIf
    \EndFor
    \Return{$D_{min}$;}
    \end{algorithmic}
    \end{algorithm}

    \vspace{-0.8cm}
    \begin{algorithm}[H] \caption{(Optimal \NGather)} \label{alg:opt_n}
    \begin{algorithmic}[1]
    \State{The robots perform Algorithm~\ref{alg:general_n} with the point $D$ determined by Algorithm~\ref{alg:opt_npoint};}
    \end{algorithmic}
    \end{algorithm}
\vspace{-0.5cm}
    \begin{theorem}\label{thm:opt_n}
    Algorithm~\ref{alg:opt_n} is an optimal algorithm solving the \NGather problem with input $\IS$. The complexity of the algorithm is $\BO(n \log n)$.
    \end{theorem}

    \ifproof
        \begin{proof}(Theorem~\ref{thm:opt_n}) By Lemmas~\ref{lm:closest_n} and \ref{lm:bisect_n} we know that the optimal point $D$ must lie on the perpendicular bisector of two points in $\IS_0$ that are furthest from $D$. In Algorithm~\ref{alg:opt_npoint} we are choosing $D$ to be on one of the edges of the FVD of $\IS_0$. Thus, by construction, $D$ does in fact lie on the perpendicular bisector of two points in $\IS_0$ that are furthest from $D$.

        The fact that the algorithm is optimal follows from Lemma \ref{lm:n_cr} and the definition of an optimal algorithm given in Eq.~\eqref{eq:opt_alg}.

        The complexity of the algorithm is $\BO(n \log n)$ since we need to determine the FVD of $\IS_0$ which can be found optimally in $\BO(n \log n)$ time \cite{Berg:2008:CGA:1370949}. Minimizing the quadratic equation given in Lemma~\ref{lm:n_cr} on a line segment can be done in constant time, and this needs to be done only once for each of the $\BO(n)$ edges of the FVD.
        \qed \end{proof}
    \fi

It does not seem likely that a closed form expression can be derived for the competitive ratio of Algorithm~\ref{alg:opt_n} for arbitrary $n$. However, in the boundary case that $n=3$ and $\fb = 1$ this is possible. The complete solution of the \TriGather is presented in the appendix and the results are reproduced below:
\begin{restatable}{theorem}{opttri} \label{thm:opt_tri}
    Algorithm~\ref{alg:general_n} optimally solves the \TriGather problem with input $\triangle{ABC}$ of side lengths $a\leq b\leq c$ and respective angles $\alpha \geq \beta \geq \gamma$ if the point $D$ is chosen such that
    $D_x = \frac{1}{2}[(B_x+C_x) + a\tan\phi (B_y-C_y)]$, 
    $D_y = \frac{1}{2}[(B_y+C_y) + a\tan\phi (C_x-B_x)]$, and
    $\tan\phi = \tan \beta$ if $\tan \beta \leq \sin \gamma$, otherwise
    \begin{equation*}
        \tan \phi = \frac{2\sqrt{c^2-(b-a)^2}}{\sqrt{(3b-a)^2-c^2} + \sqrt{(b+a)^2-c^2}}.
    \end{equation*}
    The competitive ratio of the algorithm equals $c/b$ if $\tan \beta \leq \sin \gamma$, otherwise it is $1/ \cos \phi$.
\end{restatable}
\section{Bounded number of byzantine robots}

We now consider instances of the \Gather{n}{\fb} problem when the value of $\fb$ is a small constant fraction of $n$. We give two algorithms corresponding to the cases that $\fb < \lceil \frac{n}{3} \rceil$, and $\fb < \lceil \frac{n}{2} \rceil$. In both cases we show that a small constant competitive ratio is attainable.

We start with the case that $\fb < \lceil \frac{n}{3} \rceil$. We have the following:
        \begin{theorem}\label{thm:n_3}
        Consider the \Gather{n}{\fb} problem with input $\IS$ and for any $\fb < \lceil \frac{n}{3} \rceil$. Then, there is a gathering algorithm solving this problem with competitive ratio at most 2. The complexity of the algorithm is $\BO(n)$.
        \end{theorem}

        \begin{proof}(Theorem~\ref{thm:n_3})
        We will make use of the centerpoint theorem (see \cite{edelsbrunner2012algorithms} [Theorem 4.3]) which states that any finite set $\IS$ of $n$ points in $\IR^d$ admits a point $K$ (a centerpoint) such that any open half-space avoiding $K$ contains at most $\lfloor \frac{d n}{d+1} \rfloor$ points of $\IS$. In particular, for $d=2$, this implies that we can always determine a $K$ such that any line $L$ through $K$ partitions $\IS$ into two sets each with at least $\fb < \lceil \frac{n}{3} \rceil$ robots. This result inspires the following algorithm,

        \vspace{-0.5cm}
        \begin{algorithm}[H] \caption{(Move to centerpoint)} \label{alg:cp_n3}
        \begin{algorithmic}[1]
        \State  {The robots compute a centerpoint $K$ of the set $\IS$ of robots;}
        \State  {All robots move directly towards $K$;}
        \State  {The algorithm terminates once the final reliable robot reaches $K$;}
        \end{algorithmic}
        \end{algorithm}
        \vspace{-0.5cm}

        Consider the reliable robot $A$ that is initially furthest away from the point $K$ determined in Algorithm~\ref{alg:cp_n3}. Draw a line $L$ through $K$ perpendicular to the line segment $\edge{AK}$ (as done in Figure~\ref{fig:cp_n2_n3}). Observe that, since $K$ is a centerpoint, there are at least $\lceil \frac{n}{3} \rceil$ robots on either side of $L$. Furthermore, by assumption, $\fb$ is strictly less than $\lceil \frac{n}{3} \rceil$ and we are thus guaranteed to have a reliable robot on either side of $L$. Consider any reliable robot $B$ on the opposite side of $L$ as $A$ and note that the robot $B$ is at least a distance $|AB| \geq |AK|$ away from the robot $A$. The competitive ratio of Algorithm~\ref{alg:cp_n3} is therefore at most $\CRa \leq |AK|/(\frac 12 |AB|) \leq 2$.

        The complexity bound follows from the need to determine the centerpoint of the collection. The centerpoint of a set of $n$ points can be determined in $\BO(n)$ time using an algorithm by Jadhav \cite{Jadhav1994}.\\
        \end{proof}
\vspace{-1cm}
        \includeFig{width=8cm,keepaspectratio}{fig:cp_n2_n3}{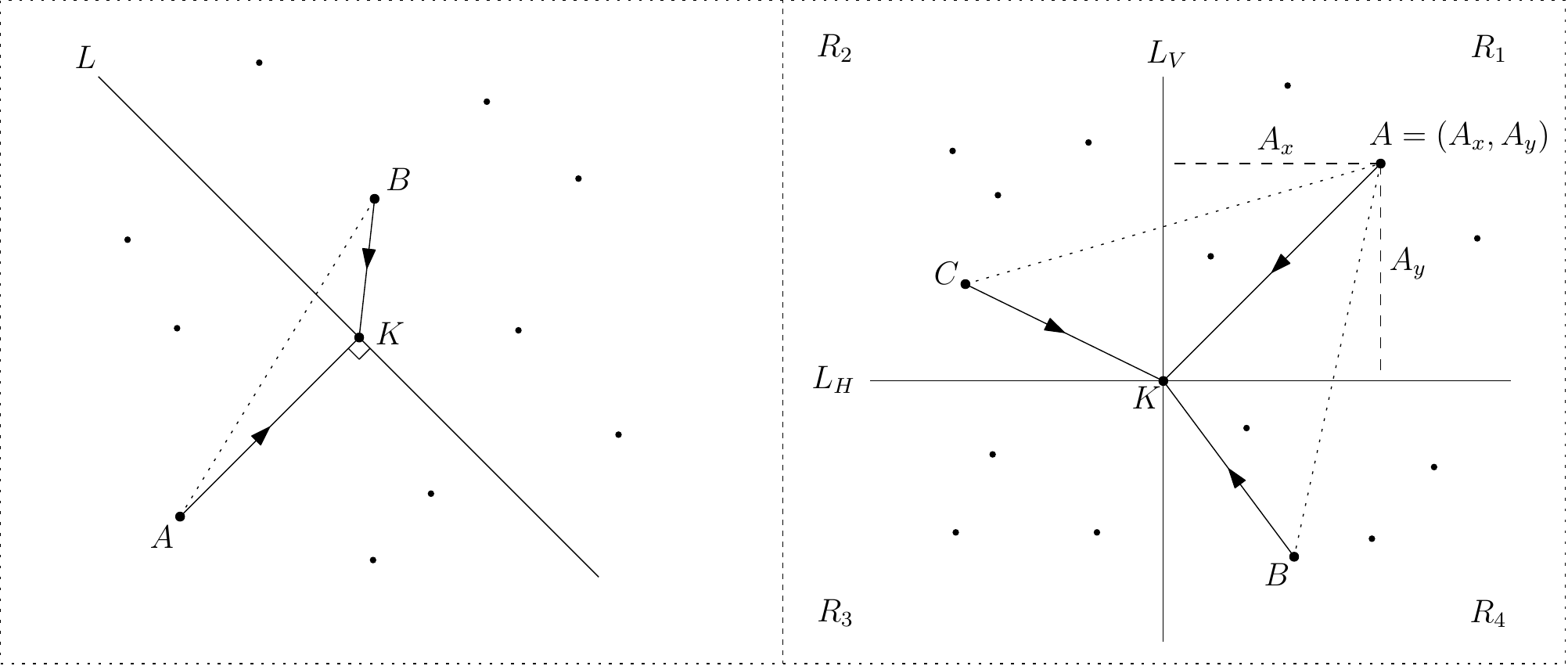}{Setup for the proofs of Theorem~\ref{thm:n_3} (left) and Theorem~\ref{thm:n_2} (right).}
    \vspace{-1cm}
    The centerpoint theorem applies generally to any $d$-dimensional space and we thus have the following corollary,

        \begin{corollary}\label{cor:n_d}(Theorem~\ref{thm:n_3}) 
        Consider the \Gather{n}{\fb} problem in $\IR^d$ for for any $\fb < \lceil \frac{n}{d+1} \rceil$. Then, there exists a gathering algorithm with competitive ratio at most 2.
        \end{corollary}

\vspace{0.3cm}

\noindent Now consider the case that $\fb < \lceil \frac{n}{2} \rceil$. We claim the following:

        \begin{theorem}\label{thm:n_2}
        Consider the \Gather{n}{\fb} problem with input $\IS$ and for any $\fb < \lceil \frac{n}{2} \rceil$. Then, there is a gathering algorithm solving this problem with competitive ratio at most $2\sqrt{2}$. The complexity of the algorithm is $\BO(n)$.
        \end{theorem}

        \begin{proof}(Theorem~\ref{thm:n_2})
        The proof is based on the following algorithm,

        \vspace{-0.5cm}
        \begin{algorithm}[H] \caption{(Move to intersection)} \label{alg:cp_n2}
        \begin{algorithmic}[1]
        \State  {The robots compute a line $L_H$ that partitions the robots into two disjoint sets each containing at least $\lceil \frac{n}{2} \rceil$ robots;}
        \State  {The robots compute a line $L_V$, perpendicular to $L_H$, that also partitions the robots into two disjoint sets each containing at least $\lceil \frac{n}{2} \rceil$ robots;}
        \State  {The robots move towards the point $K$ that is the intersection of $L_H$ and $L_V$.}
        \State  {The algorithm terminates once the final reliable robot reaches $K$;}
        \end{algorithmic}
        \end{algorithm}
        \vspace{-0.5cm}

        First, we note that, in Algorithm~\ref{alg:cp_n2}, the existence of the lines $L_H$ and $L_V$ is ensured as a result of the ham-sandwich theorem (see \cite{edelsbrunner2012algorithms} [Theorem 4.7]). 
        
        Now consider the four open regions $R_{1}$, $R_2$, $R_3$, and $R_4$ created by the intersection of $L_H$ and $L_V$ (as depicted in Figure~\ref{fig:cp_n2_n3}). Note that, by assumption, we have $\fb < \lceil \frac{n}{2} \rceil$ and we are therefore guaranteed to have at least one reliable robot in each of the regions $R_1$ and $R_3$, or in each of the regions $R_2$ and $R_4$. 
        
        Consider the reliable robot $A$ that is furthest from $K$ and assume without loss of generality that $A$ is located in the region $R_1$. If there is a reliable robot $B$ in $R_3$ then we have $|AB| \geq |AK|$ which implies that $\CRa \leq |AK|/(\frac 12 |AB|) \leq 2$. If there is not a reliable robot in $R_3$ then there must be reliable robots $B$ and $C$ in $R_2$ and $R_4$ respectively. Let $d = \max \{|AB|, |AC|\}$ and let us adopt a coordinate system such that $K = (0,0)$ and $A = (A_x, A_y)$. Observe that $A_y \leq |AB| \leq d$ and $A_x \leq |AC| \leq d$. Thus, $|AK| = \sqrt{A_x^2+A_y^2} \leq \sqrt{2} d$ and $\CRa \leq |AK|/(\frac 12 d) \leq 2\sqrt{2}$.

        The two lines $L_H$ and $L_V$ may be found in linear time by first choosing some line $L'$ onto which we project the points in $\IS$. We then set $L_H$ as the line perpendicular to $L'$ dividing the points on $L'$ in half (i.e. we need to find the median, $\BO(n)$ time \cite{BLUM1973448}). To find $L_V$ we repeat with $L'$ replaced with $L_H$.
        \end{proof}


\section{Arbitrary number of byzantine robots}

In this section we consider algorithms that solve the \Gather{n}{\fb} for any $n$ and any $\fb$. We give two algorithms: the first, grid-rendezvous, is adapted from \cite{collins2011synchronous} and gives a constant competitive ratio independent of $\fb$. The second, shrinking-the-shortest-interval (SSI), gives a competitive ratio dependent on $\fb$.

\subsection{Grid rendezvous}
    We start with the grid-rendezvous algorithm which is a direct application of Algorithm~3 in \cite{collins2011synchronous}. The algorithm was originally designed to solve the rendezvous problem of two robots unaware of the other's position (but sharing a common coordinate system).

    The idea of the algorithm is to calculate a hierarchy of grids $\Pi = \{\pi_0, \pi_1, ...\}$ which partition the plane into non-overlapping cells. The robots then travel through a series of potential meeting points located at the centers of ever larger cells from successive grids in $\Pi$.

    In detail, each $\pi_i$ exactly partitions the plane into square cells of side length $2^i$ such that one of the cells in $\pi_i$, the central cell, has its center at the origin. In order for the partition to be exact each cell is defined to include its top and right edges, as well as its top-right vertex (in addition to its interior).

    We can nearly apply Algorithm~3 as given in \cite{collins2011synchronous}. We only need to specify the finest grid division that will be used by the robots. Let $d_{\epsilon}$ be the size of this finest grid cell. We present (the slightly modified) Algorithm~3 from \cite{collins2011synchronous} below.

        \vspace{-0.5cm}
        \begin{algorithm}[H] \caption{(Grid-rendezvous \cite{collins2011synchronous})} \label{alg:gr}
        \begin{algorithmic}[1]
        \State{The robots choose a $d_{\epsilon}$ much smaller than the closest pair of robots in the set;}
        \State{The robots compute the hierarchy of grids $\Pi$;}
        \Repeat{ for $i=1,2,3...$ and for each robot in $\IS$}
        \State{Set $H$ equal to the cell of $\pi_i$ containing your initial position $p$;}
        \State{Move to the center of $H$;}
        \State{Wait until $\sqrt{2}\cdot2^{i-1}$ time has passed since the start of the current iteration;}
        \Until {Gathering completed}
        \end{algorithmic}
        \end{algorithm}
        \vspace{-0.5cm}

    The rendezvous time of the above algorithm is given by Corollary~9 in \cite{collins2011synchronous}. Using this time-bound we can state the following:

        \begin{theorem} \label{thm:gr}
        Consider the \Gather{n}{\fb} problem for the input $\IS$. Assume that the robots $A$ and $B$ are the closest pair of robots in $\IS$. Then the competitive ratio of Algorithm~\ref{alg:gr} is $\CRa \leq 2\sqrt{2} \left(16 + \frac{d_\epsilon}{|AB|} \right)$
        where $d_{\epsilon}$ can be made as small as one chooses. The complexity of this algorithm is\footnote{The complexity of the algorithm is entirely due to the determination of $d_{\epsilon}$.} $\BO(n\log n)$.
        \end{theorem}

        \begin{proof}(Theorem~\ref{thm:gr}) Choose any two robots $I$ and $J$ in $\IS$ and assume that all distances are scaled such that $d_{\epsilon} = 1$. Then Corollary~9 from \cite{collins2011synchronous} tells us that the robots $I$ and $J$ gather in time $T \leq 16\sqrt{2}\cdot|IJ| + \sqrt{2}$.

        Let $\IN$ be the subset of reliable robots and consider the two robots $A$ and $B$ which are the most distant in $\IN$. Then the minimal time necessary to gather the robots in $\IN$ is at least $T_* \geq |AB|/2$. The competitive ratio of Algorithm~\ref{alg:gr} is therefore
        $
        \CRa = \frac{T}{T_*} \leq \frac{16\sqrt{2}|AB|+\sqrt{2}}{\frac{1}{2}|AB|} \leq 2\sqrt{2}\left (16 + \frac{1}{|AB|}\right).
        $
        In the worst case, $|\IN|=2$ and the robots $A$ and $B$ were the closest pair in $\IS$.

        The complexity of the algorithm is $\BO(n \log n)$ as one needs to find the closest pair of points in $\IS$ in order to determine a satisfactory grid division $d_{\epsilon}$. The closest pair in a set of $n$ points can be found optimally in $\BO(n \log n)$ time \cite{Bentley:1980:MD:358841.358850}.
        \qed \end{proof}


\subsection{Shrink-shortest-interval}
    Consider the following algorithm, generalized from Algorithm~3 in \cite{hoda2017}:

        \vspace{-0.5cm}
        \begin{algorithm}[H] \caption{(Shrink-shortest-interval)} \label{alg:ssi}
        \begin{algorithmic}[1]
        \Repeat
        \State{Determine the two closest robots $A$ and $B$ in $\IS$ that are not at the same position;}
        \State{Set $D$ as the midpoint of $A$ and $B$;}
        \State{Set $d = |AB|/2$.}
        \State{All robots move a distance $d$ towards $D$;}
        \Until{All robots in $\IN$ gather.}
        \end{algorithmic}
        \end{algorithm}
        \vspace{-0.5cm}

    We claim the following:

        \begin{theorem}\label{thm:ssi}
        Algorithm~\ref{alg:ssi} solves the \Gather{n}{\fb} problem for the input $\IS$ with competitive ratio at most $\fb+2$. The complexity of the algorithm is $\BO(n^2 \log n)$.
        \end{theorem}

    To prove this we will need the following lemma:

        \begin{lemma}\label{lm:ssi}
        Consider any point $D$ and set of points $\IS$ such that $A\in \IS$ is the closest point to $D$, and $C \in \Sup[\IS]$ is the furthest point from $D$. Let $\IS'$ be the positions of the points in $\IS$ after moving them a distance $d \leq |AD|$ towards the point $D$. Then,
        \[\Rad[\IS'] \leq \begin{cases}
        \Rad[\IS] - d/2, & D \in \MEC(\IS)\\ 
        \Rad[\IS], & \mbox{otherwise}
        \end{cases}.\]
        \end{lemma}

        \begin{proof}(Lemma~\ref{lm:ssi}) Let $K$ and $r_{\IS}$ be the center and radius of $\MEC(\IS)$ respectively and adopt a coordinate system which places $K$ at the origin and oriented such that the line segment $\edge{KD}$ is along the positive $x$-axis. Then $D = (D_x, 0)$ and, for any point $C$ on the border of $\MEC(\IS)$ we have $C = r_{\IS} (\cos \theta,\ \sin \theta)$
        where $\theta$ is the angle between $\edge{KC}$ and the $x$-axis. 
        
        Define the point $C'$ as the point obtained by moving the point $C$ a distance $d\leq|CD|$ towards $D$. We can write
        $C'_x = r_{\IS} \cos\theta + \frac{d}{|CD|}(D_x - r_{\IS}\cos\theta) = r_{\IS}\cos\theta \left(1-\frac{d}{|CD|} \right) + \frac{d\cdot D_x}{|CD|}$, and, 
        $C'_y = r_{\IS} \sin\theta - \frac{d}{|CD|}\cdot r_{\IS}\sin\theta = r_{\IS}\sin\theta \left(1-\frac{d}{|CD|} \right)$.       
        Observe that the equations for $C'_x$, and $C'_y$ describe a parametric curve that is completely contained within a circle of radius $r_{\IS}(1-d/|CD|) \leq r_{\IS}$. Thus, for any point $D$, we can conclude that $\Rad[\IS']\leq r_{\IS}(1-d/|CD|) \leq r_{\IS}$.
        
        Now consider the case that $D$ is inside $\MEC(\IS)$. In this case $|CD| \leq 2r_{\IS}$ such that the curve defined by $C'_x$, and $C'_y$ is completely contained within a circle of radius $r_{\IS}(1-d/|CD|) \leq (r_{\IS}-d/2)$.
        \qed \end{proof}

        \begin{proof}(Theorem~\ref{thm:ssi})
        Consider the \Gather{n}{\fb} problem for the input $\IS$, let $\IN$ be the subset of $\IS$ that contains only reliable robots, and let $f$ be the (actual) number of byzantine robots in $\IS$. 
        
        Let $\IS^{(i)}$ and $\IN^{(i)}$ represent the unique positions of the robots in $\IS$ and $\IN$ after the $i^{th}$ iteration of the algorithm, and let $r_i = \Rad[\IN^{(i)}]$. We also let $D_i$ be the midpoint and $d_i$ be half the distance between the closest pair of points in $\IS^{(i)}$. Finally, set $C_i \in \Sup[\IN^{(i)}]$ be the furthest point from $D_i$.

        Now, if in the $i^{th}$ iteration the midpoint $D_i$ lies within $\MEC(\IN^{(i)})$ then by Lemma~\ref{lm:ssi} we have $r_{i+1} \leq r_{i} - d_i/2$. If we assume that their are $m$ iterations of this kind then the time needed to complete these iterations is at most
        $T_m \leq \sum_{i=0}^{m}d_i \leq 2\sum_{i=0}^{m}(r_i-r_{i+1}).$
        However, observe that
        $\sum_{i=0}^{m} r_i = r_0 + \sum_{i=1}^{m} r_i = r_0 + \sum_{i=0}^{m-1} r_{i+1}$
        such that $T_m \leq 2 r_0 - r_{m+1} \leq 2 r_0$.

        If $D_i$ does not lie within $\MEC(\IN^{(i)})$, then we can only say that $r_{i+1} \leq r_{i}(1-d_i/|C_i D_i|) \leq r_{i}$. However, observe that Algorithm~\ref{alg:ssi} always gathers the two closest robots in $\IS^{(i)}$ and we know that there is at least one pair of robots in $\IN^{(i)}$ with separation no greater than $2r_i$. This tells us that $d_i \leq r_i$. Furthermore, since all reliable robots are, by definition, within $\MEC(\IN)$, it is impossible for $D_i$ to simultaneously be: a) the midpoint of two reliable robots, and, b) lie outside of $\MEC(\IN)$. This implies that this type of iteration can occur at most $f$ times (as it reduces the number of byzantine robots by one each time it occurs). Thus, the time needed to complete these iterations is at most $T_f = f \cdot r_0$.

        Combining $T_m$ and $T_f$ gives us a bound on the total time necessary to complete the algorithm. We get $T \leq T_m + T_f = f r_0 + 2r_0 = (f+2)r_0$.
        The bound on the competitive ratio follows from the fact that $f \leq \fb$, and $r_0 = \Rad[\IN]$ is the minimal time necessary to gather the robots in $\IN$.

        The complexity bound follows from the fact that we need to determine the closest pair of points $\BO(n)$ times.
        \qed \end{proof}

    In the case that we have no knowledge of the number of byzantine robots in our collection (i.e. $\fb = n-2$) the algorithm has a worst-case bound on the competitive ratio of $n$. This reflects the fact that an adversary, if allowed, would always choose $f = \fb$ robots in $\IS$ to be byzantine. It is worth noting, however, that it was not necessary to know $\fb$ in the proof of Theorem~\ref{thm:ssi} and thus the algorithm has a competitive ratio that is bounded by the actual number of byzantine robots in $\IS$. That is, for a particular instance $\IN \subseteq \IS$ such that $f = |\IS| - |\IN|$ we have $\CR(\IN) \leq f+2 \leq \fb+2$.

\section{Conclusion}
    In this paper we analyzed the gathering problem for $n>2$ robots in
    the plane at most $\fb$ of which, $\fb \leq n-2$, are byzantine.
    The robots were equipped with GPS and they could communicate their
    positions to a central authority. Several algorithms were designed
    with competitive ratio depending on the number of byzantine robots and the
    knowledge available to the robots. 
    
    In addition to improving the competitive ratio and/or complexity of our algorithms, several
    interesting open problems remain. In particular, one could consider
    models that allow the robots to communicate/exchange their positions
    at any time during the gathering process. Additionally, it would be
    interesting to consider robot gathering (in the presence of byzantine
    robots) under local (limited) communication range.

\bibliographystyle{plain}
\bibliography{bibliography}

\appendix
\section{Gathering three robots}

Observe that we can describe an instance of the \TriGather problem by the triangle $\triangle{ABC}$ whose vertices specify the initial positions of the three robots for which gathering should occur. Thus, throughout this section, we let $a = |BC|$, $b = |AC|$, and $c = |AB|$ be the side lengths of $\triangle{ABC}$, and set the angles $\beta = \angle {ABC}$, and $\gamma = \angle {BCA}$. Without loss of generality we assume that $a \leq b \leq c$.

With this description of the problem Lemmas~\ref{lm:closest_n}, \ref{lm:bisect_n}, \ref{lm:n_cr}, and Theorem~\ref{thm:n_lb} take the following simple forms:

\begin{corollary}(Lemma~\ref{lm:closest_n} and Lemma~\ref{lm:bisect_n}) An optimal algorithm solving the \TriGather problem with input $\triangle{ABC}$ has the robots $B$ and $C$ meet first at some point on their perpendicular bisector.
\end{corollary}

\begin{corollary}(Theorem~\ref{thm:n_lb} and Lemma~\ref{lm:n_cr}) An optimal algorithm solving the \TriGather problem with input $\triangle{ABC}$ has competitive ratio
\begin{equation}\label{eq:tri_cr}\CRa = \max \left\{ \frac{2|BD|}{a},\ \frac{|AD|+|BD|}{b}\right\}\end{equation}
and this is at least $c/b$.
\end{corollary}

At this point we could simply apply Algorithm~\ref{alg:opt_n} to determine the optimal point for this problem, however, in order to derive a closed form expression for the point $D$ we take a different approach. We claim the following:

\opttri*

    \begin{proof}(Theorem~\ref{thm:opt_tri}) First consider the case that $\tan \beta \leq \sin \gamma$ (see Figure~\ref{fig:tri_leq}, top). In this case $D$ lies on the edge $\edge{AB}$ and we have $2|BD|/a = 1/\cos\beta$ and $|BD|+|AD| = |AB| = c$. Thus, by Eq.~\eqref{eq:tri_cr}, we have, $\CRa = \max \left \{ \frac{1}{\cos\beta},\ \frac{c}{b} \right\}$.
    We claim that the condition $\tan \beta \leq \sin \gamma$ implies that $1/\cos\beta \leq c/b$. Indeed, observe that
    $\tan \beta \leq \sin\gamma$ which we can rewrite as $\frac{1}{\cos \beta} \leq \frac{\sin\gamma}{\sin \beta}=\frac{c}{b}$ (where we have invoked the sine law in the last step).

    
    \begin{center}
        \vspace{-0.2cm}
        \begin{minipage}[t][4.5cm]{.55\linewidth}
        \vspace{0pt}
        \centering
        \includegraphics[width=6cm,keepaspectratio]{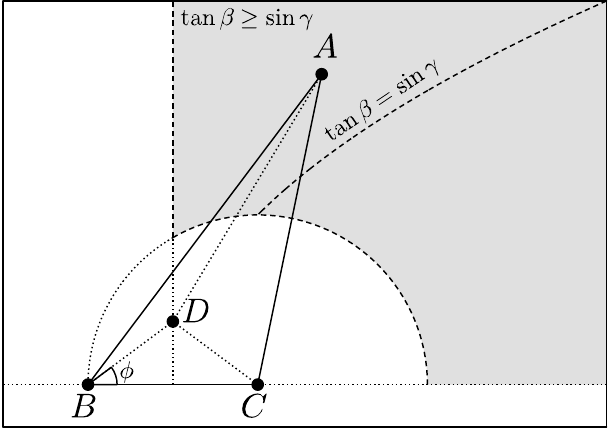}
        \end{minipage}%
        \begin{minipage}[t][4.5cm]{.45\linewidth}
        \vspace{0pt}
        \centering
        \includegraphics[width=6cm,keepaspectratio]{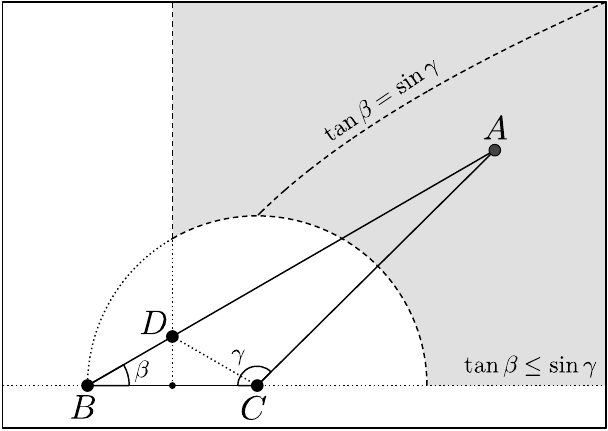}
        \end{minipage}
        \captionof{figure}{Setup for the proof of Theorem~\ref{thm:opt_tri}. The shaded gray region indicates those positions of $A$ such that $a\leq b\leq c$. Right: $\tan \beta \leq \sin \gamma$. Left: $\tan \beta \geq \sin \gamma$}
        \label{fig:tri_leq}
        \end{center}

    Now consider the case that $\tan \beta \leq \sin \gamma$ (see Figure~\ref{fig:tri_leq}, bottom) and define $\CR_A = \frac{2|BD|}{a}$, and $\CR_B = \frac{|BD|+|AD|}{b}$ such that $\CRa = \max \{\CR_A,\ \CR_B\}$.
    We note that, over the interval $[0, \beta]$, $CR_B$ is a monotone continuous decreasing function of $\phi$ and $CR_A$ is a monotone continuous increasing function of $\phi$. Furthermore, $CR_A \geq CR_B$ when $\phi = \beta$ and $CR_A = 1 \leq CR_B$ when $\phi = 0$. We can thus conclude that the optimal competitive ratio is given when $CR_A = CR_B$, and the optimal point $D$ is determined from the following:
    \begin{equation}\label{eq:tri0}
    \frac{2|BD|}{a} = \frac{|BD|+|AD|}{b}.
    \end{equation}

    Now, let us choose a coordinate system such that the midpoint of $B$ and $C$ is at the origin, and the positive $y$-axis lies along the perpendicular bisector of $B$ and $C$. In this coordinate system $B=(-a/2, 0)$, $C = (a/2, 0)$, $A=(A_x, A_y)$, and $D = (0, d)$ where we have introduced the parameter $d$. Thus, we can write
    $|BD| = \sqrt{\frac{1}{4}a^2+d^2}$, and $|AD| = \sqrt{\left(A_x-\frac{a}{2}\right)^2+A_y^2}$.
    Plugging these into Eq.~\eqref{eq:tri0} gives us (after some manipulation) the following quadratic equation
    $ 4b(b-a) d^2 + 2a\Delta d - \frac{a^2}{2}[c^2-(b-a)^2] = 0 $
    where $\Delta$ is the area of $\triangle{ABC}$. Solving for $d$ we get (after some manipulation),
    \[ d = \frac{a[c^2-(b-a)^2]}{\sqrt{4\Delta^2 + 2b(b-a)(c^2-(b-a)^2)} + 2\Delta}. \]
    Finally, by applying Heron's formula for the area of a triangle, and noting that $d = \frac a2 \tan \phi$, we get our final result:
    \[\tan \phi = \frac{2\sqrt{c^2-(b-a)^2}}{\sqrt{(3b-a)^2-c^2} + \sqrt{(b+a)^2-c^2}}.\]
\qed \end{proof}

It is interesting to note that in some cases ($\tan \beta \leq \sin \gamma$) it was possible to achieve the lower-bound given in Theorem~\ref{thm:n_lb}. This turns out to be true for the general case as well, and, in fact, it is possible to specify under which conditions this occurs:

    \begin{lemma}\label{lm:achieve_lb}
    Consider an input $\IS$ such that $|Sup[\IS]| = 2$. Let $L$ be the line segment defined by the two points in $Sup[\IS]$, and let $\mathcal{P}$ be the convex region defined by the intersection of all circles with centers given by the points in $\IS_0$ and radii given by $r_0 \cdot r_{\IS}/r_1$. Then, if $L$ intersects with $\mathcal{P}$, the competitive ratio of Algorithm~\ref{alg:opt_n} is $\CRa = r_{\IS}/r_1$.
    \end{lemma}

    \begin{proof}(Lemma~\ref{lm:achieve_lb}) Define the convex region $\mathcal{P}$ as above and observe that $\mathcal{P} \neq \emptyset$ since, at minimum, it contains the center of $\MEC(\IS_0)$.

    Now, we let $A$ in $\IS_0$ and $C \in \IS_1$ be the two points that are also in $\Sup[\IS]$.  Observe that for any point $D$ on the segment $L = \edge{AC}$ we have $|AD| + |CD| = |AC| = 2r_{\IS}$. Thus, if $A$ is (one of) the furthest point(s) from $D$, the algorithm terminates in at most $r_{\IS}$ time. 
    
    Now observe that for any point $D \in \mathcal{P}$ the distance between $D$ and any point in $\IS_0$ is at most $q$. In particular, for the point $A$, we have $|AD| \leq q$.

    Thus, if it happens that: a) $\mathcal{P}$ and $L$ intersect, and, b) we can choose $D$ in $\mathcal{P} \cap L$ in such a way that $A$ is (one of) the furthest point(s) from $D$, then, the competitive ratio of Algorithm~\ref{alg:general_n} is

    \begin{align*}
    \CRa = \max \left \{ \frac{|AD|}{r_0},\ \frac{|AD|+|CD|}{2r_1} \right\}
                    \leq \max \left \{ \frac{q}{r_0},\ \frac{2r_{\IS}}{2r_1} \right\} \leq \max \left \{ \frac{r_{\IS}}{r_1},\ \frac{r_{\IS}}{r_1} \right\} = \frac{r_{\IS}}{r_1}.
    \end{align*}

    Now, observe that there must be a point in $\mathcal{P} \cap L$ such that $A$ is furthest from it. One such point, for example, is the point of intersection of the segment $L$ and the circle of radius $q$ centered on $A$. In order to get the optimal point we simply need to choose the $D$ that minimizes the quadratic given in Lemma~\ref{lm:n_cr} on the portion of $L$ contained in the region of $\FVD(\IS_0)$ belonging to the point $A$.
    \qed \end{proof}

\end{document}